\documentclass[12pt]{article}
\usepackage{amsmath}
\usepackage{enumerate}
\usepackage[round]{natbib} %
\usepackage{url} 

\usepackage{booktabs}
\usepackage{amssymb,amsthm}

\newcommand{\var}{\mathrm{var}}
\newcommand{\cor}{\mathrm{cor}}
\newcommand{\eff}{\mathrm{effar}}
\newcommand{\leff}{\mathrm{leffar}}

\newcommand{\efford}{\mathrm{eff}}

\newcommand{\natu}{\mathbb{N}}

\newcommand{\kmax}{k_{\max}}

\newcommand{\logonep}{\mathrm{log1p}}
\newcommand{\rd}{\,\mathrm{d}}
\renewcommand{\ge}{\geqslant}
\renewcommand{\le}{\leqslant}

\newtheorem{theorem}{Theorem}
\newtheorem{lemma}{Lemma}
\newtheorem{proposition}{Proposition}

\newcommand{\blind}{0}

\newcommand{\orho}{\bar \rho}
\newcommand{\urho}{\underline\rho}

\begin{document}

\def\spacingset#1{\renewcommand{\baselinestretch}%
{#1}\small\normalsize} \spacingset{1}


\if0\blind
{
  \title{\bf Statistically efficient thinning of a Markov chain sampler}
  \date{This version: April 2017}
  \author{Art B. Owen\thanks{
    This work was supported by NSF grants
\textit{DMS-1407397 and DMS-1521145.}}
\hspace{.2cm}\\
Department of Statistics\\ 
Stanford University}
  \maketitle
} \fi
\bigskip
\begin{abstract}
It is common to subsample Markov chain output to  
reduce the storage burden.  
\cite{geye:1992} shows that discarding $k-1$ out of every $k$ observations  
will not improve statistical efficiency, as quantified through
variance in a given computational budget.
That observation is often taken to mean that thinning
MCMC output cannot improve statistical efficiency.
Here we suppose that it costs one unit of time to advance
a Markov chain and then $\theta>0$ units of time to compute
a sampled quantity of interest. 
For a thinned process, that cost $\theta$ is incurred less
often, so it can be advanced through more stages.
Here we provide examples to show that thinning
will improve statistical efficiency if $\theta$ is large
and the sample autocorrelations decay slowly enough.
If the lag $\ell\ge1$ autocorrelations of a scalar measurement
satisfy $\rho_\ell\ge\rho_{\ell+1}\ge0$, then there is always a $\theta<\infty$
at which thinning becomes more efficient for averages of that scalar.
Many sample autocorrelation functions resemble
first order AR(1) processes with $\rho_\ell =\rho^{|\ell|}$
for some $-1<\rho<1$.
For an AR(1) process it is possible
to compute the most efficient subsampling frequency $k$.  
The optimal $k$ grows rapidly as $\rho$ increases towards $1$.  
The resulting efficiency gain depends primarily on $\theta$, not $\rho$.  
Taking $k=1$ (no thinning) is optimal when $\rho\le0$. For  
$\rho>0$ it is optimal if and only if  
$\theta \le (1-\rho)^2/(2\rho)$.  This efficiency gain never exceeds $1+\theta$.  
This paper also gives efficiency bounds
for autocorrelations bounded between those of two AR(1) processes.
\end{abstract}

\noindent%
\vfill

\section{Introduction}
It is common to thin a Markov chain sample,
taking every $k$'th observation instead of
all of them.  Such subsampling is done to produce
values that are more nearly independent.  It also
saves storage costs.  It is well known that the average
over a thinned sample set has greater variance than
the plain average over all of the computed values  
\citep{geye:1992}.

Most authors recommend against thinning, except where it
is needed to reduce storage.
\cite{mace:berl:1994} go so far as to provide a 
`justification for the ban against subsampling'.
\cite{link:eato:2011} 
write that ``Thinning is often unnecessary and always inefficient''.
In discussing thinning of the Gibbs sampler,
\cite{gamerman2006markov} 
say: ``There is no gain in efficiency, however, by this approach and
estimation is shown below to be always less precise than retaining
all chain values.''

One exception is \cite{geye:1991} who acknowledges that
thinning can in fact increase statistical efficiency.
Thinning reduces the average cost of iterations which then makes it possible
to run a thinned Markov chain longer than an unthinned one
at the same computational cost.
He gives some qualitative remarks about this effect, but ultimately concludes
that it is usually a negligible benefit because the autocorrelations in
the Markov chain decay exponentially fast. \cite{link:eato:2011}
also acknowledge this possibility in their discussion
as does \citet[page 106]{neal:1993}.

This paper revisits the thinning problem and shows that the usual
advice against thinning can be misleading, by quantifying the argument
of \cite{geye:1991} described above.  The key variables are
the cost of computing the quantity of interest (after advancing the Markov
chain) and the speed at which correlations in the quantity of interest decay.
When the cost is expensive and the decay is slow, then thinning can
improve efficiency by a large factor.

We suppose that it costs one unit to advance the
Markov chain and $\theta>0$ units each time the quantity
of interest is computed.  If lag $\ell$ autocorrelations 
satisfy $\rho_1\ge\rho_2\ge\cdots>0$, then there is always
a $\theta$ for which thinning by a factor of $k$ will
improve efficiency.

For a first order autoregressive autocorrelation structure in the quantity of interest, 
very precise results are possible.  
Given the update costs  and the autocorrelation parameter we can  
compute the optimal thinning factor as well as the efficiency  
improvement with that factor.  
The autoregressive assumption is very convenient because it reduces the  
dependence problem to just one scalar parameter.  
Also, real-world  autocorrelations commonly resemble those  of an AR(1) model.  
In the social sciences,
the book by \cite{jack:2009} shows many sample autocorrelation  
functions that resemble AR(1).   
The physicists~\cite{newm:bark:1999} writing about the Ising model
state that ``the autocorrelation is expected to fall off exponentially
at long times'' (p 60).
\cite{geye:1991} notes an exponential upper bound for autocorrelations
when processes are $\rho$-mixing.

Sometimes  thinning is built in to standard
simulation practice. For instance an Ising model may be simulated
as a sequence of `passes' with each pixel being examined on average
once per pass.  The state of the Markov chain might only
be inspected once per pass. That represents
a substantial, though not necessarily optimal amount of thinning.
It might really be better to sample several times per pass or
just once every $k$ passes.

An outline of this paper is as follows.
Section~\ref{sec:asyeff} defines asymptotic efficiency
of thinning to every $k$'th observation when the samples
have unit cost to generate, the function of interest costs $\theta>0$
each time we compute it.
If the autocorrelations $\rho_\ell$ for $\ell\ge1$ are nonnegative and nonincreasing
and $\rho_k>0$ then there is always some finite $\theta>0$
for which thinning by a factor of $k$ is more efficient
than not thinning.
Much sharper results can be obtained
when the autocorrelations take the form $\rho_\ell=\rho^\ell$ at lag~$\ell$. 
In many cases the optimal thinning factor $k$ is greater than one.

Section~\ref{sec:sometheory} presents some inequalities
among the efficiency levels at different subsampling frequencies
in the AR(1) case.
Thinning never helps when $\rho\le0$.
For $\rho>0$, if any thinning level is to help, then taking every second
sample must also help, and as a result we can get sharp expressions for
the autocorrelation level at which thinning increases efficiency. In the
limit $\rho\to1$ very large thinning factors become optimal but
frequently much smaller factors are nearly as good. The efficiency gain
does not exceed $1+\theta$ for any $\rho$ and $k$.
Section~\ref{sec:nearauto} considers autocorrelations that are
bounded between two autoregressive forms $\urho^\ell\le\rho_\ell\le\orho^\ell$.
The range of optimal thinning factors widens, but it is often possible
to find meaningful efficiency improvements from thinning.
Section~\ref{sec:optimization} describes how to compute the optimal
thinning factor $k$ given the parameters $\theta$ and an autoregression parameter $\rho$.
Section~\ref{sec:conclusions} has conclusions and discusses
consequences of rejected proposals having essentially zero cost
while accepted ones have a meaningfully large cost.
An appendix has R code to compute the optimal $k$.

We close with some practical remarks. When thinning benefits,
it does not appear to be critical to find the optimal factor $k$.  Instead
there are many near optimal thinning factors.  If the autocorrelations decay
slowly and the cost $\theta$ is large then a suggestion of Hans
Anderson is to thin in such a way that about half of the cost is spent
advancing the Markov chain and about half is spent computing the
quantity of interest.  That should be nearly as efficient as using the optimal $k$.

\section{Asymptotic efficiency}\label{sec:asyeff}
To fix ideas, suppose that we generate a Markov chain
$x_t$ for $t\ge1$. We have a starting value $x_0$ and 
then it costs one unit
of computation to transition from $x_{t-1}$ to $x_t$. 
The state space for $x_t$ can be completely general
in the analysis below.

Interest centers on the
expected value of $y_t = f(x_t)$ for some real-valued function $f$.
There is ordinarily more than one such function, but here we
focus on a single one.
The cost to compute $f$ is $\theta$.  Often $\theta\ll1$
but it is also possible that $\theta$ is comparable to $1$
or even larger.  
For instance it may be inexpensive to perform one
update on a system of particles, but very expensive to find the
new minimum distance among all those particles or some similar
quantity of interest. Or, it may be very inexpensive to flip one
or more edges in a simulated network but expensive to compute
a connectivity property of the resulting network. Finally, when
computation must pause to store $f(x_t)$,
then the cost of pausing is included in $\theta$.

The efficiency of thinning derived here depends
on the cost of computing $y_t$ from $x_t$, the
cost of transition from $x_t$ to $x_{t+1}$, and the autocovariances of
the series $y_t$. We assume that $y_t$ is stationary:  any necessary
warmup has taken place.

The variance of $\sqrt{n}\hat\mu\equiv(1/\sqrt{n})\sum_{i=1}^nf(x_i)$ is asymptotically
$\sigma^2(1+2\sum_{\ell=1}^\infty\rho_\ell)$
where $\rho_\ell =\cor(y_i,y_{i+\ell})$ and $\sigma^2 = \var(y_i)$.
We assume that $0<\sigma^2<\infty$.
Now suppose that we thin the chain as follows. We compute $y_i=f(x_i)$
only for every $k$'th observation.  
The number of function values we get will depend on $k$.
If we take $n_k$ of them then we estimate $\mu$ by
$$
\hat\mu_k = \frac1{n_k}\sum_{i=1}^{n_k}f(x_{ik}).
$$
To compute $\hat\mu_k$ we must advance the chain $kn_k$ times
and evaluate $f$ at each of $n_k$ points for a total cost of $n_k(k+\theta)$. 
When our computational budget is a cost of $B>0$, then
we will use the largest $n_k$ with $n_k(k+\theta)\le B$. That is $n_k = \lfloor B/(k+\theta)\rfloor$.

The relative efficiency of thinning by a factor $k$ compared to not thinning
at all is
$$
\efford_B(k) =
\frac{
 (\sigma^2/n_1)(1+2\sum_{\ell=1}^\infty\rho_\ell) 
}{
 (\sigma^2/n_k)(1+2\sum_{\ell=1}^\infty\rho_{k\ell}) 
}
=
\frac{\lfloor B/(k+\theta)\rfloor}{\lfloor B/(1+\theta)\rfloor}
\frac{
1+2\sum_{\ell=1}^\infty\rho_\ell 
}{
1+2\sum_{\ell=1}^\infty\rho_{k\ell}
}. 
$$
The dependence on $B$ is minor and is a nuisance.  We work instead with
\begin{align}\label{eq:efford}
\efford(k) =
\frac{1+\theta}{k+\theta}
\frac{
1+2\sum_{\ell=1}^\infty\rho_\ell 
}{
1+2\sum_{\ell=1}^\infty\rho_{k\ell}
},
\end{align}
which is also the limit of $\efford_B(k)$ as $B\to\infty$.

\subsection{Generic autocorrelations}

The efficiency of thinning depends on the autocorrelations $\rho_\ell$
only through certain sums of them.
We can use this to get inequalities on autocorrelations
that are equivalent to statements on the efficiency $\efford(k)$.
Then under a monotonocity constraint on autocorrelations
we can get a condition that ensures that thinning will help.

\begin{lemma}\label{lem:rmklow}
Let $R=\sum_{\ell=1}^\infty\rho_\ell$, and for a thinning factor  
$k\ge1$, define $R_k=\sum_{\ell=1}^\infty\rho_{k\ell}$
and $R_{-k} = R-R_k$.  
Then $\efford(k)<1$ if and only if
\begin{align}\label{eq:rmklow}
R_{-k} < \frac{k-1}{\theta+1}\bigl(R_k+1/2\bigr).  
\end{align}
\end{lemma}
\begin{proof}
We rewrite~\eqref{eq:efford} as
$$
\efford(k)=
\frac{1+\theta}{k+\theta}\frac{1+2R_k+2R_{-k}}{1+2R_k}.  
$$ 
Then  $\efford(k)<1$ if  and only if
$1+2R_k+2R_{-k}<(1+2R_k)(k+\theta)/(1+\theta),$ 
which can be rearranged into~\eqref{eq:rmklow}.
\end{proof}

Only one out of every $k$ consecutive autocorrelations 
contributes to $R_k$ while the other $k-1$ of them contribute to $R_{-k}$. 
If we let $\bar R_{-k}=R_{-k}/(k-1)$, then equation~\eqref{eq:rmklow}
becomes $\bar R_{-k} < (R_k+1/2)/(\theta+1)$. 
For a Markov chain with slowly converging autocorrelations
we will have $R_k\gg 1/2$. Then
for thinning to be inefficient, the autocorrelations
contributing to $R_k$ have to be enough larger than the others
to overcome the factor $\theta+1$.  When $\theta$ is large
we would then need every $k$'th autocorrelation to be surprisingly
large compared to the nearby ones, in order to make thinning inefficient.

Now suppose that the autocorrelations satisfy
\begin{align}\label{eq:monorho}
\rho_1\ge\rho_2\ge\cdots\ge0. 
\end{align}
This quite mild sufficient condition rules out a setting where every
$k$'th autocorrelation is unusually large compared to
its $k-1$ predecessors.

\begin{theorem}\label{thm:monorho}
If~\eqref{eq:monorho} holds then $\efford(k)<1$ can only hold
for $\theta < 1/(2R_k).$
\end{theorem}
\begin{proof}
From Lemma~\ref{lem:rmklow}, $\efford(k)<1$ implies that
$R_{-k} <(R_k+1/2)(k-1)/(\theta+1)$.
If~\eqref{eq:monorho} holds then $(k-1)R_k\le R_{-k}$.
Therefore
$(k-1)R_{k} < (R_k+1/2)(k-1)/(\theta+1)$ which can 
be rearranged to complete the proof.
\end{proof}

If $\rho_\ell$ are large and slowly decreasing, then $R_k$
will be quite large and $1/(2R_k)$ will be very small.
Then even for mild costs $\theta$, Theorem~\ref{thm:monorho}
ensures that some form of thinning will improve asymptotic
efficiency.  
The converse does not hold:  thinning might still
help, even if $\theta <1/(2R_k)$.

The condition~\eqref{eq:monorho} includes the
case with $\rho_\ell=0$ for all $\ell>1$. This is a case
where thinning cannot help. We also
get $R_k=0$ here, so Theorem~\ref{thm:monorho} then
places no constraint on~$\theta$, consistent with the
fact that thinning cannot then help.
If~\eqref{eq:monorho} holds, then all we need is 
$\rho_k>0$ to get $R_k>0$. Then there is
a $\theta<\infty$ for which $\efford(k)>1$ holds.

\subsection{AR(1) autocorrelations}

Here we consider a
first order autoregressive model, $\rho_k = \rho^{k}$ for $\rho\in(-1,1)$ and $k\in\natu$. In this setting it is possible to find the most efficient values
of $k$ and to measure the efficiency gain from them.
It is reasonable to expect qualitatively similar results from autocorrelations that
have approximately the AR(1) form.
Some steps in that direction are in Section~\ref{sec:nearauto}.

Under an AR(1) model
\begin{align}\label{eq:eff}
\efford(k) = \eff(k)
= \eff(k;\theta,\rho)\equiv
\frac {1+\theta}{k+\theta}
\frac{1+\rho}{1-\rho}
\frac{1-\rho^k}{1+\rho^k}. 
\end{align}
We use $\eff(k)$ to denote an efficiency computed under the autoregressive assumption
and $\efford(k)$ to denote a more general efficiency.
The efficiency in~\eqref{eq:efford} is a continuous function of the underlying $\rho_\ell$ inside of it,
so small departures from the autoregressive assumption will make small changes in
efficiency.  When the peak of $\efford(k)$ is flat then small changes in $\rho_\ell$ may
bring large changes in $\arg\max_k\efford(k)$.

Table~\ref{tab:bestk} shows $\arg\max_k\eff(k;\rho,\theta)$ for
a range of correlations $\rho$ and costs $\theta$. 
This $k$ is computed via a search 
described in Section~\ref{sec:optimization}. As one would
expect, the optimal thinning factor increases with both $\theta$ and $\rho$.

Perhaps surprisingly, the optimal thinning factor can be large even for $\theta \ll1$,
when the chain mixes slowly.  For instance with $\theta = 0.01$ 
and $\rho = 0.9999$, the optimal thinning takes every $182$'nd value.
But Table~\ref{tab:effbestk} shows that in such cases only 
a small relative efficiency gain occurs.  For $\theta = 0.01$ and
$\rho=0.9999$ the improvement is just under $1$\% and this gain may not be worth the trouble
of using thinning.

When the cost $\theta$ is comparable to one, then thinning can bring a meaningful
efficiency improvement for slow mixing chains.  The efficiency gain approaches $\theta+1$
in the limit as $\rho\to1$.
See equation~\eqref{eq:limeffk} in Section~\ref{sec:sometheory}.

\begin{table}[p]\centering 
\begin{tabular}{lrrrrrrrr}
\toprule 
$\theta\ \backslash\ \rho$&      0.1& 0.5& 0.9 &0.99 &0.999 &0.9999 &0.99999 &0.999999\\
\midrule 
0.001 &  1 &  1 &  1 &   4 &   18 &    84 &    391 &    1817\\
0.01 &   1 &  1 &  2 &   8  &  39 &   182 &    843 &    3915\\
0.1 &    1 &  1 &  4 &  18  &  84 &   391 &   1817 &    8434\\
1 &      1 &  2 &  8  & 39 &  182 &   843 &   3915 &   18171\\
10 &     2 &  4 & 17 &  83 &  390 &  1816 &   8433 &   39148\\
100 &    3 &  7 & 32 & 172 &  833 &  3905 &  18161 &   84333\\
1000 &   4 & 10 & 51 & 327  &1729 &  8337 &  39049 &  181612\\
\bottomrule 
\end{tabular}
\caption{\label{tab:bestk}
Optimal thinning factor $k$ as a function of the relative cost 
$\theta$ of function evaluation and the autoregressive parameter $\rho$. 
}
\end{table}

\begin{table}\centering 
\begin{tabular}{lrrrrrrrr}
\toprule 
$\theta\ \backslash\ \rho$&      0.1& 0.5& 0.9 &0.99 &0.999 &0.9999 &0.99999 &0.999999\\
\midrule 
0.001 &1.00 &1.00 & 1.00 &  1.00 &  1.00 &  1.00 &   1.00 &    1.00 \\
0.01 & 1.00 &1.00 & 1.00 &  1.01 &  1.01 &  1.01 &   1.01 &    1.01  \\
0.1 &  1.00 &1.00 & 1.06 &  1.09   &1.10 &  1.10 &   1.10 &    1.10  \\
1 &    1.00 &1.20 & 1.68  & 1.93   &1.98  & 2.00 &   2.00 &    2.00  \\
10 &   1.10 &2.08 & 5.53  & 9.29  &10.59 & 10.91 &  10.98&    11.00  \\
100 &  1.20 &2.79 &13.57  &51.61 &  85.29  &97.25&  100.17&   100.82  \\
1000 & 1.22 &2.97 &17.93 &139.29 &512.38 &845.38&  963.79 &  992.79  \\
\bottomrule 
\end{tabular}
\caption{\label{tab:effbestk}
Asymptotic efficiency of the optimal thinning factor $k$ 
from Table~\ref{tab:bestk} as a function of 
$\theta$ and $\rho$.  Values rounded to two places.
}
\end{table}

A more efficient thinning rule allows the user to wait
less time for an answer, or to attain a more accurate answer
in the same amount of time.   It may be a slight nuisance
to incorporate thinning and when storage is not costly,
we might even prefer to explore a larger set of sampled $y$ values.
Table~\ref{tab:okk} shows the least amount of thinning that
we can do to get at least $95$\% efficiency relative to the
most efficient value of $k$. That is, we find the smallest $k$
with $\eff(k;\rho,\theta)\ge0.95\min_{\,\ell\ge1}\eff(\ell;\rho,\theta)$. 
When $95$\% efficiency is adequate and $\theta$ is small
then there is no need to thin.
Theorem~\ref{thm:effbound} below shows that in the AR(1) model,
there is no
need to thin at any $\rho$, if efficiency $1/(1+\theta)$ is acceptable.

\begin{table}\centering 
\begin{tabular}{lrrrrrrrr}
\toprule 
$\theta\ \backslash\ \rho$&      0.1& 0.5& 0.9 &0.99 &0.999 &0.9999 &0.99999 &0.999999\\
\midrule 
0.001 &  1 &  1 &  1 &   1 &    1 &     1 &      1 &       1 \\
0.01 &   1 &  1 &  1 &   1 &    1 &     1 &      1 &       1 \\
0.1 &    1 &  1 &  2 &   2 &    2 &     2 &      2 &       2 \\
1 &      1 &  2 &  5 &  11 &   17 &    19 &     19 &      19 \\
10 &     2 &  4  &12 &  45 &  109  &  164  &   184 &     189 \\
100 &    2 &  5  &22 & 118 &  442 &  1085 &   1632 &    1835 \\
1000 &   2 &  6  &31 & 228 & 1182 &  4415 &  10846  &  16311 \\
\bottomrule 
\end{tabular}
\caption{\label{tab:okk}
Smallest $k$ to give at least $95$\% of the efficiency 
of the most efficient $k$, as a function of 
$\theta$ and the autoregression parameter $\rho$. 
}
\end{table}

\section{Some inequalities}\label{sec:sometheory}

Here we compare efficiencies for different choices
of the thinning factor $k$, under the autoregressive assumption
$\rho_\ell = \rho^\ell$.
We find that thinning never helps when $\rho<0$.
In the limit as $\rho\to1$, the optimal $k$ diverges to
infinity but we can attain nearly full efficiency by taking $k$
to be a modest multiple of $\theta$.
When $\rho>0$, the critical value of $\theta$, meaning  
one large enough to make  
$\eff(k;\rho,\theta)>\eff(1;\rho,\theta)$, 
is an increasing function of $k\ge2$.  
As a result we can determine when $k=1$ is optimal.
The following basic lemma underpins several of the results.

\begin{lemma}\label{lem:basic}
Let $r>s\ge1$ be integers, $\theta\ge0$ and $-1<\rho<1$.
Then $\eff(r;\rho,\theta)>\eff(s;\rho,\theta)$ if and only if
\begin{align}\label{eq:basic}
2(\theta+s)(\rho^s-\rho^r)>(r-s)(1-\rho^s)(1+\rho^r).
\end{align}
\end{lemma}
\begin{proof}
Because $(1+\theta)(1+\rho)/(1-\rho)>0$, the given inequality in efficiencies is equivalent to
$$
(s+\theta)(1-\rho^r)(1+\rho^s) > (r+\theta)(1-\rho^s)(1+\rho^r).
$$
Equation~\eqref{eq:basic} follows by rearranging this inequality.
\end{proof}

It is obvious that thinning cannot improve efficiency when
$\rho=0$. 
Here we find that the same holds for $\rho<0$.

\begin{proposition}
If $\theta\ge0$ and $-1<\rho\le 0$ then
$\eff(1;\rho,\theta)\ge\eff(k;\rho,\theta)$ for all 
integers $k\ge2$.
\end{proposition}
\begin{proof}
Suppose to the contrary that $\eff(k)>\eff(1)$. Then
Lemma~\ref{lem:basic} with $r=k$ and $s=1$ yields
\begin{align}\label{eq:rearranged}
2(\theta+1)(\rho-\rho^k)&>(k-1)(1-\rho)(1+\rho^k).
\end{align}
Because the right side of~\eqref{eq:rearranged} is positive and the left side is not, we
arrive at a contradiction, proving the result.
\end{proof}

With negative $\rho$ there is an advantage to taking an odd integer $k$
compared to nearby even integers, stemming from the factor 
$(1+\rho^k)/(1-\rho^k)$ in $\eff(k)$.
For instance with Lemma~\ref{lem:basic} we find that $\eff(3;\rho,\theta)>\eff(2;\rho,\theta)$
when $\rho<0$, but $k=1$ remains the best odd integer.  From here on we restrict
attention to $\rho>0$.
Also it is obvious that $k=1$ is best when $\theta=0$ and so we assume $\theta>0$.

Many applications have correlations very close to $1$. 
Then 
\begin{align}\label{eq:limeffk}
\eff(k;1,\theta)\equiv\lim_{\rho\to1}\eff(k;\rho,\theta) = 
k\frac {1+\theta}{k+\theta}. 
\end{align}
The optimal $k$ grows without bound as $\rho\to1$
and it has asymptotic efficiency $\theta+1$. 
From Tables~\ref{tab:bestk} and~\ref{tab:effbestk} we might anticipate
that there are diminishing returns to very large $k$ in this limit.
If we do not insist on maximum efficiency we can use much smaller $k$.
To obtain efficiency at least $1-\eta$ relative to the best $k$  in the large $\rho$
limit we impose
$$ \eff(k;1,\theta) \ge (1+\theta)(1-\eta)$$
for $0<\eta<1$.
Rearranging this inequality we obtain
$$
k\ge\theta(1-\eta)/\eta.
$$
For instance to attain $95$\% efficiency relative to the best $k$
in the large $\rho$ limit,
we may take $\eta=0.05$ and then $k = \lceil 19\theta\rceil$.

The next proposition introduces a critical cost $\theta_*(k)$
beyond which thinning by the factor $k$ is more efficient than
not thinning.  That threshold cost increases with $k$ and the result is
that we may then characterize when $k=1$ (no thinning) is optimal.

\begin{proposition}\label{prop:criticaltheta}
Let $0<\rho<1$ and choose an integer $k\ge2$. 
Then $\eff(k;\rho,\theta) >\eff(1;\rho,\theta)$ if and only if 
\begin{align}\label{eq:crittheta}
\theta > \theta_*(k,\rho) \equiv  \frac{k-1}2\frac{(1-\rho)(1+\rho^k)}{\rho-\rho^k}-1. 
\end{align}
\end{proposition}
\begin{proof}
This follows from Lemma~\ref{lem:basic} using $r=k$ and $s=1$. 
\end{proof}

For fixed $\rho\in(0,1)$ and very large $k$ we
find that
$$\theta_*(k,\rho)\doteq\frac{(k-1)(1-\rho)}{2\rho}-1.$$
If $\rho$ is near zero, then choosing large $k$ is only efficient for
very large $\theta$. If $\rho$ is close to $1$ then the threshold
for $k$'s efficiency can be quite low.

\begin{proposition}\label{prop:thetamonotoneink}
For $0<\rho<1$ the critical $\theta_*(k,\rho)$ from~\eqref{eq:crittheta}
is an increasing function of $k\ge 2$.
\end{proposition}
\begin{proof}
Let $r=k-1\ge1$ and put 
$$\phi_*(r,\rho) = \frac{2(\theta_*+1)\rho}{1-\rho}=
r\frac{1+\rho^{r+1}}{1-\rho^{r}}.
$$
It suffices to show that $\phi_*$ is increasing over $r\in[1,\infty)$.
Differentiating,
\begin{align*}
\frac{\rd\phi_*}{\rd r}
&=
\frac{(1+\rho^{r+1}+r\rho^{r+1}\log(\rho))(1-\rho^r)
-r(1+\rho^{r+1})(-\rho^r\log(\rho))
}
{(1-\rho^r)^2}
\end{align*}
and we need only show that the numerator 
%
$$
\eta(\rho,r) = 1-\rho^r+\rho^{r+1}-\rho^{2r+1}+(\rho^r+\rho^{r+1})\log(\rho^r)
$$
is positive.
We will show that $\eta(\rho,r)\ge\eta(\rho,1)\ge0.$

First we show that $\eta(\rho)\equiv\eta(\rho,1)\ge0$. This function and its 
first three derivatives are 
\begin{align*}
\eta(\rho) &= 1-\rho+\rho^{2}-\rho^{3}+(\rho+\rho^{2})\log(\rho),\\
\eta'(\rho) &= 3\rho-3\rho^2+(1+2\rho)\log(\rho),\\
\eta''(\rho) &= 5-6\rho+\rho^{-1}+2\log(\rho),\quad\text{and}\\
\eta'''(\rho) &= -6-\rho^{-2}+2\rho^{-1}\\
& = -5-(1-\rho^{-1})^2. 
\end{align*} 
Because $\eta'''\le0$ and $\eta''(1-)=0$ we find that $\eta''(\rho)\ge0$. 
Similarly, $\eta'(1-)=0$ and so $\eta'(\rho)\le0$. 
Finally, $\eta(1-)=0$ so that $\eta(\rho)=\eta(\rho,1)\ge0$,
completing the first step. 

For the second step, treating $r\ge1$ as a continuous variable,
\begin{align*}
\frac{\partial}{\partial r}
\eta(\rho,r) &= -\rho^r\log(\rho)+\rho^{r+1}\log(\rho) 
-2\rho^{2r+1}\log(\rho)\\
&\quad+(\rho^r+\rho^{r+1})\log(\rho) 
+r(\rho^r+\rho^{r+1})\log^2(\rho)\\
&=2(\rho^{r+1}-\rho^{2r+1})\log(\rho)+r(\rho^r+\rho^{r+1})\log^2(\rho)\\
&=\rho^r\log(\rho) \bigl(2\rho-2\rho^{r+1}+r\log(\rho)+r\rho\log(\rho)\bigr). 
\end{align*}
We now show that this partial derivative is nonnegative.
Because $\rho^r\log(\rho)\le0$, it suffices to show that the second factor 
$F(\rho,r)\le0$ where 
$F(\rho,r)\equiv2\rho-2\rho^{r+1}+r\log(\rho)+r\rho\log(\rho)$. 
Differentiating yields 
\begin{align*}
\frac{\partial}{\partial \rho} F(\rho,r) 
&= 2-2(r+1)\rho^r+r\rho^{-1}+r\log(\rho)+r,\quad\text{and}\\
\frac{\partial^2}{\partial \rho^2} F(\rho,r) 
&= -2r(r+1)\rho^{r-1}+r\rho^{-1}(1-\rho^{-1})\le0. 
\end{align*}
Proceeding as before, $(\partial/\partial\rho)F(1-,r)=0$ so 
this first partial derivative is nonnegative. Then $F(1-,r)=0$ so $F$ is nonpositive 
as required. 
\end{proof}

\begin{theorem}
For $0<\rho<1$, the choice $k=1$ maximizes efficiency $\eff(k;\rho,\theta)$
over integers $k\ge1$ whenever 
\begin{align}\label{eq:thetabound}
\theta \le\frac{(1-\rho)^2}{2\rho}.
\end{align}
For $\theta>0$, the choice $k=1$ maximizes efficiency $\eff(k;\rho,\theta)$ if
\begin{align}\label{eq:critrhofor1}
\rho \le 1+\theta -\sqrt{\theta^2+2\theta}.
\end{align}
\end{theorem}
\begin{proof}
From the monotonicity of $\theta_*$ in Proposition~\ref{prop:thetamonotoneink},
if any $k>1$ is better than $k=1$ then $\eff(2;\rho,\theta)>\eff(1;\rho,\theta)$.
Then $k=1$ is most efficient if
$$\theta\le\theta_*(2,\rho) = 
\frac{(1-\rho)(1+\rho^2)}{2(\rho-\rho^2)}-1= \frac{(1-\rho)^2}{2\rho},$$
establishing~\eqref{eq:thetabound}.
The equation $\theta = (1-\rho)^2/(2\rho)$ has two roots $\rho$ 
for fixed $\theta>0$
and~\eqref{eq:thetabound} holds for $\rho$ 
outside the open interval formed by those two roots.
One of those roots is larger than $1$ and the other is given as the right hand side of~\eqref{eq:critrhofor1}.
\end{proof}

The upper limit in~\eqref{eq:critrhofor1} is asymptotically
$1/(2\theta)$ for large $\theta$.
That is $\lim_{\theta\to\infty} \theta
(1+\theta -\sqrt{\theta^2+2\theta})=1/2$.

\begin{theorem}\label{thm:effbound}
For integer lag $k\ge1$, cost $\theta>0$ and autocorrelation
$0<\rho<1$, the function $\eff(k;\rho,\theta)$ is nondecreasing in $\rho$,
and so
$$
\eff(k;\rho,\theta)\le\theta+1.
$$
\end{theorem}
\begin{proof}
The second conclusion follows from the first using
the limit in~\eqref{eq:limeffk}.
The derivative of $\eff(k;\rho,\theta)\times(k+\theta)/(1+\theta)$ 
with respect to $\rho$ is
\begin{align}\label{eq:dedrho}
\frac{(1-k\rho^{k-1}-(k+1)\rho^k)(1-\rho)(1+\rho^k)
-(1+\rho)(1-\rho^k)(-1+k\rho^{k-1}-(k+1)\rho^k)
}
{(1-\rho)^2(1+\rho^k)^2}.
\end{align}
It suffices to show that the numerator in~\eqref{eq:dedrho} is non-negative
for $0<\rho<1$.
The numerator simplifies to twice 
$N(\rho,k) = k\rho^{k+1}-k\rho^{k-1}-\rho^{2k}+1$.
Now 
\begin{align*}
\frac{\partial}{\partial\rho}N(\rho,k) 
&=k(k+1)\rho^k-k(k-1)\rho^{k-2}-2k\rho^{2k-1}=k\rho^{k-2}F(\rho,k) 
\end{align*} 
for a factor $F(\rho,k)=(k+1)\rho^2-(k-1)-2\rho^{k+2}$. 
Because $F(1-,k)=0$ and $\partial F(\rho,k)/\partial\rho 
=2(k+1)(\rho-\rho^k)\ge0$ we have $F(\rho,k)\le0$. 
Therefore $\partial N(\rho,k)/\partial\rho\le0$ and 
because $N(1-,k)=0$ we conclude that $N(\rho,k)\ge0$
and so $\eff(k;\rho,\theta)$ is nondecreasing in $\rho$.
\end{proof}

Next we consider how to locate the maximizer over $k$ of $\eff(k;\rho,\theta)$.
The next result on relaxing $\log(k)$ to be a nonnegative real value helps.

\begin{proposition}\label{prop:itslogconcave}
For $\theta>0$ and $\rho\in(0,1)$ the function
$\log( \eff(e^x;\rho,\theta))$ is strictly concave in $x\ge0$.
\end{proposition}
\begin{proof}
We write $ \log( \eff(e^x;\rho,\theta)) = g(x)+f(x)$ for 
\begin{align*}
g(x)&=\log\Bigl(\frac{1+\theta}{e^x+\theta}\Bigr)+\log\Bigl(\frac{1+\rho}{1-\rho}\Bigr),\quad\text{and}\quad 
f(x)=\log\Bigl(\frac{1-\rho^{e^x}}{1+\rho^{e^x}}\Bigr). 
\end{align*}
Now $g''(x) = -\theta e^x/(e^x+\theta)^2$, so $g$ is strictly concave for $\theta>0$. It now suffices 
to show that $f$ is concave. 
Let $h=h(x)=\rho^{e^x}$. Then 
\begin{align*}
f'(x) &= \frac{-h'}{1-h}-\frac{h'}{1+h} = \frac{-2h'}{1-h^2}= \frac{-2h\log(h)}{1-h^2},
\end{align*}
using $h'=h\log(h)$, and so 
\begin{align}\label{eq:halfcurv}
-\frac12f''(x) &= \frac{
(h'\log(h)+h')(1-h^2)-h\log(h)(-2hh') 
}{(1-h^2)^2}. 
\end{align}
The numerator in~\eqref{eq:halfcurv} is 
$h\log(h)\bigl(1-h^2+\log(h)+h^2\log(h)\bigr)$. 
We need only show that $t(z)=1-z^2+\log(z)+z^2\log(z)\le0$ on $(0,1]$ which 
includes all relevant values of $h$. 
The result follows because $t(1)=t'(1)=t''(1)=0$ and $t'''(z)=2z^{-3}+2z^{-1}\ge0$. 
\end{proof}

From Proposition~\ref{prop:itslogconcave},
we know that if we relax $k$ to real values in $[1,\infty)$, then
$\eff(k;\rho,\theta)$ is either nonincreasing as $k$ increases
from $1$, or it increases to a peak 
 before descending, or it increases indefinitely as $k$ increases.
Because $\lim_{k\to\infty}\eff(k;\rho,\theta)=0$ we can rule out the third possibility.

As a result, we know that if $\eff(k';\rho,\theta)<\eff(k;\rho,\theta)$
for $k'>k$, then the optimal value of $k$ is in the set $\{1,2,\dots,k'-1\}$.
Neither the value $k'$ nor any larger value can be optimal because the
function $\eff(k;\rho,\theta)$ is already in its descending region by the
time we consider lags as large as $k'$.

\section{Approximately autoregressive dependence}\label{sec:nearauto}

In this section we consider how efficiencies and optimal thinning factors
behave when the autocorrelations are nearly but not exactly
of autoregressive form.
Recall that the efficiency of thinning to every $k$'th observation versus
not thinning ($k=1$) is given by $\efford(k)$ of~\eqref{eq:efford}.
More generally, the efficiency of thinning by factor $r\in\natu$ versus thinning by
factor $s\in\natu$ is
$$\efford(r,s) = \frac{\efford(r)}{\efford(s)}=
\frac{s+\theta}{r+\theta}
\frac{
1+2\sum_{\ell=1}^\infty\rho_{s\ell} 
}{
1+2\sum_{\ell=1}^\infty\rho_{r\ell}
}.
$$

Thinning should help for autocorrelations that are approximately autoregressive
and  decay very slowly.  Then the numerator in $\efford(k,1)$ will be large.
For there to be a meaningful efficiency gain, the denominator in $\efford(k,1)$
should not be too large.  That is reasonable as we ordinarily expect that $\rho_{k\ell} <\rho_\ell$.
We suppose now that the autocorrelations of $y_i$ satisfy
\begin{align}\label{eq:bounds}
\urho^\ell \le \rho_\ell \le \orho^\ell .
\end{align}
Then $\rho_\ell$ need not follow exactly the autoregressive form and indeed
they need not be monotonically decreasing in $\ell$.

Under condition~\eqref{eq:bounds} we can get upper and lower bounds on the
summed autocorrelations and these yield
\begin{align*}
\frac{s+\theta}{r+\theta}
\frac{1+\urho^{s}}{1-\urho^{s}}\Bigm/
\frac{1+\orho^{r}}{1-\orho^{r}}
\le 
\efford(r,s)
\le 
\frac{s+\theta}{r+\theta}
\frac{1+\orho^{s}}{1-\orho^{s}}\Bigm/
\frac{1+\urho^{r}}{1-\urho^{r}}\equiv U_{rs}. 
\end{align*}
Any value $r$ for which $U_{rs}<1$ holds for some $s\ne r$
cannot be the optimal thinning factor.  

There are two main ways that thinning can help.  One is that the
autocorrelations decay slowly. The other is that the cost $\theta$ 
to compute $y_i=f(x_i)$ is large.  We consider one example of each
type.

First, consider a slow but not extremely slow correlation decay,  $\urho=0.98$ and $\orho=0.99$ with moderately large $\theta=10$.
If $U_{1k}<1$, then thinning at factor $k$ must be more efficient than not thinning for
any autocorrelation satisfying~\eqref{eq:bounds}. 
We find that this holds whenever $3\le k\le 1078$.
If $U_{1k}<1/2$, then thinning at factor $k$ must be at least twice as efficient as not thinning.
This holds for $6\le k\le 529$.  If $28\le k\le 195$ then thinning is at least four
times as efficient as not thinning.  In this not very extreme example, there are gains from
thinning and they hold over a wide range of thinning factors $k>1$.
The thinning factors that are not dominated
by some other thinning factor are given by $8\le k\le 220$. Any other $k$ cannot be optimal.
The given values of $\urho$, $\orho$ and $\theta$ allow for a large
set of possible optimal $k$, but they do not allow for $k=1$ to be optimal.  Instead, $k=1$
is suboptimal by at least four-fold.

As a second example, consider a high cost $\theta=100$ with moderately slow correlation
decay given by $\urho = 0.9$ and $\orho =0.95$.  Then there is at least a $10$-fold efficiency
gain for any $34\le k\le 87$ and 
the optimal $k$ must satisfy $16\le k\le 74$.

\section{Optimization}\label{sec:optimization}

The most direct way to maximize $\eff(k;\rho,\theta)$ over
$k\in\natu$ is to compute $\eff(k;\rho,\theta)$ for all $k=1,\dots,\kmax$
and then choose 
 $$k_*=k_*(\rho,\theta) = \arg\max_{1\le k\le \kmax}\eff(k;\rho,\theta).$$
It is necessary to find a value $\kmax$ that we can be sure is 
at least as large as $k_*$.  
In light of the discussion following the log concavity Proposition~\ref{prop:itslogconcave},
we need only find a value
$\kmax$ where $\eff(\kmax;\rho,\theta)<\eff(k';\rho,\theta)$ holds
for some $k'<\kmax$.  We do this by repeatedly doubling $k$ until
we encounter a decreased efficiency.


For moderately large values of $\theta$ and $1/(1-\rho)$ it is
numerically very stable to compute $\eff(k;\rho,\theta)$.  But for
more extreme cases it is better to work with
\begin{align*}
\leff(k) \equiv \log(\eff(k;\rho,\theta) )&= 
c(\rho,\theta) - \log(k+\theta) 
+\log(1-\rho^k)-\log(1+\rho^k),
\end{align*}
where $c(\rho,\theta) = \log[ (1+\theta)(1+\rho)/(1-\rho)]$ does not
depend on $k$.
Many computing environments contain a special function $\logonep(x)$
that is a numerically more precise way to compute $\log(1+x)$ for
small $|x|$.
Ignoring $c$ we then work with
\begin{align*}
\leff'(k) \equiv 
- \log(k+\theta) 
+\logonep(-\rho^k)-\logonep(\rho^k).
\end{align*}

Now, to find $\kmax$ we set
$m=1$ and then
while  $\leff'(2m) > \leff'(m)$ set $m = 2m$.    
At convergence take $\kmax=2m$.
R code to implement this optimization is given in the Appendix.
Only in extreme circumstances will $\kmax$ be larger than
one million, and so the enumerative approach will ordinarily have a trivial cost
and it will not then be necessary to use more sophisticated searches.
It takes about $1/6$ of a second for this search to produce the values in 
Tables~\ref{tab:bestk} and~\ref{tab:effbestk} on a MacBook Air.
If $\kmax$ is thought to be extraordinarily large then one could 
run a safeguarded Newton method to find $x_*=\arg\max_x\log(\eff(e^x;\rho,\theta))$
and whichever of $k=\lceil e^{x_*}\rceil$ or $k=\lfloor e^{x_*}\rfloor$ maximizes $\eff(k;\rho,\theta)$.

\section{Discussion}\label{sec:conclusions}

Contrary to common recommendations,
thinning a Markov chain sample can improve 
statistical efficiency.  This phenomenon always holds
for monotonically decreasing nonnegative autocorrelations
if the cost of evaluating $f$ is large enough.
When the correlations follow
an autoregressive model, the optimal subsampling rate 
grows rapidly as $\rho$ increases towards $1$
becoming unbounded in the limit. Sometimes 
those large subsampling rates correspond to only modest 
efficiency improvements. 
The magnitude of the improvement depends greatly on 
the ratio $\theta$ of the cost of function evaluation to 
the cost of updating the Markov chain. When $\theta$ is 
of order $1$ or higher, a meaningful efficiency improvement 
can be attained by thinning such a Markov chain. 
When the autocorrelations decay slowly but do not necessarily
follow the exact autoregression pattern we may still find that
thinning brings a large efficiency gain. 

In some problems, the cost $\theta$ may have an important
dependence on $k$.
In an MCMC, it is common
to have $x_{t+1}=x_t$ because a proposal was rejected.
In such cases $f(x_{t+1})=f(x_t)$ need not be recomputed.
Then an appropriate cost measure for $\theta$ would be the 
CPU time taken to evaluate $f$, normalized by the time to
generate a proposal, and then multiplied by the acceptance rate.
Larger values of $k$ increase the chance that a proposal
has been accepted and hence the average cost of computing $f$.
For instance, \cite{gelman1996efficient}
find that an acceptance rate of $\alpha=0.234$ is most efficient in high
dimensional Metropolis random walk sampling.
Then when thinning by factor $k$, the appropriate cost
is $\theta(1-\alpha^k)$ where $\theta$ is the cost of an accepted proposal
and the efficiency becomes 
$$
\frac{1+\theta(1-\alpha)}{k+\theta(1-\alpha^k)}
\frac{1+\rho}{1-\rho}
\frac{1-\rho^k}{1+\rho^k}
$$
under an autoregressive assumption.
Optimizing this case is outside the scope of this article.
It is more difficult because the autocorrelation $\rho$ depends on 
the acceptance rate $\alpha$. At any level of thinning, the optimal $\alpha$
may depend on $\theta$.

It is also common that one has multiple functions $f_1,\dots,f_M$
to evaluate. They might each have different optimal thinning
ratios. Optimizing the efficiency over such a collection raises
issues that are outside the scope of this article.  For instance,
the cost of evaluating a subset of these functions may be
subadditive in the costs of evaluating them individually due
to shared computations. The importance of estimating those
$M$ different means may also be unequal. Finally, there may be
greater statistical efficiency for comparisons of those
corresponding means when the $f_j$ are evaluated on common
inputs.

\section*{Acknowledgments}
This work was supported by the NSF under
grants DMS-1407397 and DMS-1521145.
I thank Hera He, Christian Robert, 
Hans Andersen, Michael Giles and some anonymous reviewers
for helpful comments.

\bibliographystyle{apalike}
\bibliography{thinning}

\vfill\eject
\section*{Appendix: R code}
\begin{verbatim}
# Code to find the optimal amount of thinning for a Markov sample.
# It costs 1 unit to advance the chain, and theta units to evaluate 
# the function. The autocorrelation is rho.

effk = function(k,theta,rho){
# Asymptotic efficiency of thinning factor k vs using k=1
# Compute and exponentiate log( effk )
# NB: log1p( x ) = log( 1+x )
t1 = log1p(theta) - log(k+theta)
t2 = log1p(rho) - log1p(-rho)
t3 = log1p(-rho^k) - log1p(rho^k)

exp( t1 + t2 + t3 )
}

leffkprime = function(k,theta,rho){
# Log of asymptotic efficiency at thinning factor k.
# It ignores terms that do not depend on k.

if( any( rho!=0 ) & any( abs(rho^k) == 0) ){
# Basic detection of underflow while still allowing rho=0
  badk = min( k[abs(rho^k)==0] )
  msg = paste("Underflow for k >=",badk,sep=" ")      
  stop(msg)
}
- log(k+theta) + log1p(-rho^k) - log1p(rho^k)
}

getkmax = function(theta,rho){
# Find an upper bound for the optimal thinning fraction k
if( theta<0 )stop("Negative theta")
if( rho<0 )stop("Negative rho")
if( rho >=1 )stop("rho too close to one")

m=1
while( leffkprime(2*m,theta,rho) > leffkprime(m,theta,rho) )
  m = m*2    
2*m
}

kopt = function(theta,rho,klimit=10^7){
# Find optimal k for the given theta and rho.
# Stop if kmax is too large. That usually
# means that theta is very large or rho is very nearly one

kmax = getkmax(theta,rho)    
if( kmax > klimit ){
  msg = paste("Optimal k too expensive. It requires checking",kmax,"values.")
  stop(msg)
}
leffvals = leffkprime( 1:kmax,theta,rho )
best = which.max(leffvals)
best
}

kok = function(theta,rho,klimit=10^7,eta=.05){
# Find near optimal k for the given theta and rho.
# This is the smallest k with efficiency >= 1-eta times best.
# NB: computations in kopt are repeated rather than
# saved. This is inefficient but the effect is minor.

best = kopt(theta,rho,klimit)
leffvals = leffkprime( 1:best,theta,rho )
ok = min( which(leffvals >= leffvals[best]  + log1p(-eta) ) )
ok
}

kopttable = function( thvals = 10^c(-3:3), rhovals = c(.1,.5,1-10^-c(1:6)),eta=.05){

# Prepare tables of optimal k, its efficiency, and smallest 
# k with at least 1-eta efficiency

T = length(thvals)
R = length(rhovals)

bestk = matrix(0,T,R)
row.names(bestk) = thvals
colnames(bestk) = rhovals
effbk = bestk
okk = bestk

for( i in 1:T )
for( j in 1:R ){
  theta = thvals[i]
  rho   = rhovals[j]
  bestk[i,j] = kopt(theta,rho)
  effbk[i,j] = leffkprime(bestk[i,j],theta,rho)-leffkprime(1,theta,rho)
  effbk[i,j] = exp(effbk[i,j])
  okk[i,j]   = kok(theta,rho,eta=eta)
}

list( bestk=bestk, effbk=effbk, okk=okk )
}
\end{verbatim}

\end{document}